\documentclass[onecolumn]{IEEEtran}

\usepackage{graphicx}
\usepackage{textcomp}
\usepackage{etex}
\usepackage{dcolumn}
\usepackage{epsfig}
\usepackage[cmex10]{amsmath}
\usepackage{amssymb}
\usepackage{amsthm}
\usepackage{amsfonts}
\usepackage{subcaption}
\usepackage{setspace}
\usepackage[pagebackref=false]{hyperref}
\usepackage{bm}
\usepackage{cite}
\usepackage[svgnames]{xcolor} 
\usepackage{pstricks,pst-node,pst-plot,pstricks-add}
\usepackage{authblk}
\usepackage{algpseudocode}
\usepackage[noabbrev]{cleveref}

\usepackage{algorithm}

\usepackage{tikz}

\usepackage{mathdots}
\usepackage{yhmath}
\usepackage{cancel}
\usepackage{color}
\usepackage{siunitx}
\usepackage{array}
\usepackage{multirow}
\usepackage{gensymb}
\usepackage{tabularx}
\usepackage{booktabs}
\usepackage{xspace}

\input{Definition.def}
\graphicspath{{figs/}}
\newcommand{\figref}[1]{Fig.\,\ref{#1}}
\newcommand{\tabref}[1]{Table~\ref{#1}}
\tikzset{every picture/.style={line width=0.75pt}}
\def\BibTeX{{\rm B\kern-.05em{\sc i\kern-.025em b}\kern-.08em
		T\kern-.1667em\lower.7ex\hbox{E}\kern-.125emX}}

 \newcommand{\light}{\texttt{LightSecAgg}\xspace}
\newcommand{\google}{\texttt{SecAgg}\xspace}
\newcommand{\googlep}{\texttt{SecAgg+}\xspace}
\newcommand{\turbo}{\texttt{TurboAgg}\xspace}

\newcommand{\swift}{\texttt{SwiftAgg}\xspace}

\begin{document}
	\title{\swift: Communication-Efficient and Dropout-Resistant Secure Aggregation for Federated Learning with Worst-Case Security Guarantees}
	\author[1,4]{Tayyebeh Jahani-Nezhad}
	\author[1]{Mohammad Ali Maddah-Ali}
	\author[2,3]{Songze Li}
	\author[4]{Giuseppe Caire}
	\affil[1]{Department of Electrical Engineering, Sharif University of Technology}
	\affil[2]{IoT Thrust, The Hong Kong University of Science and Technology (Guangzhou)}
	\affil[3]{Department of Computer Science and Engineering, The Hong Kong University of Science and Technology}
	\affil[4]{Electrical Engineering and Computer Science Department, Technische Universit\"at Berlin}
	
	\renewcommand\Authands{ and }
	\maketitle
	
	\begin{abstract}
		
We propose \swift,  a novel secure aggregation protocol for federated learning systems, where a central server aggregates local models of $N$ distributed users, each of size $L$, trained on their local data, in a privacy-preserving manner. Compared with state-of-the-art secure aggregation protocols, \swift significantly reduces the communication overheads without any compromise on security. Specifically, in presence of at most $D$ dropout users,  \swift achieves a users-to-server communication load of $(T+1)L$ and a users-to-users communication load of up to $(N-1)(T+D+1)L$, with a worst-case information-theoretic security guarantee, against any subset of up to $T$ semi-honest users who may also collude with the curious server. The key idea of \swift is to partition the users into groups of size $D+T+1$, then in the first phase,  secret sharing and aggregation of the individual models are performed within each group, and then in the second phase, model aggregation is performed on $D+T+1$ sequences of users across the groups. If a user in a sequence drops out in the second phase, the rest of the sequence remain silent. 
This design allows only a subset of users to communicate with each other, and only the users in a single group to directly communicate with the server, eliminating the requirements of 1) all-to-all communication network across users; and 2) all users communicating with the server, for other secure aggregation protocols. This helps to substantially slash the communication costs of the system.
	\end{abstract}
	
	
	\section{Introduction}
	Federated learning (FL) is an emerging distributed learning framework that allows a group of distributed users (e.g., mobile devices) to collaboratively train a global model with their local private data, without sharing the data~\cite{mcmahan2017communication,kairouz2019advances,li2020federated}. 
	Specifically, in a FL system with a central server and $N$ users, during each training iteration, the server sends the current state of the global model to the users. Receiving the global model, each user then trains a local model with its local data, and sends the local model to the server. By aggregating the local models, the server can update the global model for the next iteration. 
	While the local datasets are not directly shared with the server, several studies have shown that 
	a curious server can launch model inversion attacks to reveal information about the training data of individual users from their local models (see, e.g.,~\cite{zhu2020deep,geiping2020inverting}). Therefore, the key challenge to protect users' data privacy is to design \emph{secure aggregation} protocols, which allow the aggregation of local models to be computed without revealing each individual model. Moreover, as some users may randomly drop out of the aggregation process (due to low batteries or unstable connections), the server should be able to robustly recover the aggregated local models of the surviving users, in a privacy-preserving manner. 
	
As such motivated, 
a secure aggregation protocol \google was proposed in~\cite{bonawitz2017practical}, where each user's local model is hidden under masks computed from pair-wise random seeds. These masks have an additive structure and can be canceled out when aggregated at the server, hence the exact model aggregation can be recovered without compromising each user's data privacy. To deal with user dropouts, each user secret shares its private seed 
with other users, such that the pair-wise masks between dropped and surviving users can be reconstructed at the server, and removed from the final aggregation result.
One of the major challenges for \google to scale up is the communication cost. First, the secret sharing among users requires all-to-all communication, which incurs quadratic cost in the number of users $N$ and is often not even feasible in practical scenarios; second, every user has to communicate its masked model to the server, yielding significant communication latency at the server as $N$ increases (especially for deep models with hundreds of millions of parameters).



There has been a series of works that aim to improve the communication efficiency of \google~(see, e.g.,~\cite{so2021turbo,bell2020secure,choi2020communication}). In~\cite{so2021turbo}, \turbo was proposed to perform secure aggregation following a circular topology, achieving a communication cost of $\mathcal{O}(LN\log N)$ at the server and $\mathcal{O}(L\log N)$ at each user, for a model size $L$.
\googlep was proposed in \cite{bell2020secure} to consider a $k$-regular communication graph among users instead of the complete graph, where $k=\mathcal{O}(\log N)$. It is shown that \googlep requires a communication of $\mathcal{O}(LN+ N\log N)$ at the server and $\mathcal{O}(L+ \log N)$ at each user. 
In~\cite{choi2020communication}, another similar idea was proposed in which a sparse random graph is used as communication network instead of the complete graph. 
While these approaches improve the communication efficiency of \google, they only provide probabilistic privacy guarantees as opposed to the worst-case guarantee of \google. Besides communication, other secure aggregation protocols have been proposed to reduce the computation complexity of \google~\cite{kadhe2020fastsecagg,yang2021lightsecagg}.

In this paper, we propose a new scheme for secure aggregation in federated learning called \swift, which reduces the communication loads and is robust against user dropouts.
In \swift, we first partition the users into groups of size $T+D+1$ (See Fig.~\ref{fig3}). Then in the first phase, users within each group secret share their local models and aggregate the shares locally. In the second phase, $T+D+1$ sequences of users are arranged, such that in each sequence, there is one user from each group. Then, in-group aggregated shares are sequentially aggregated in each sequence, which is finally sent to the server. In the second phase, if one user in a sequence drops out, the rest of the sequence remain silent.  

\swift simultaneously achieves the following advantages compared to the existing works:
	\begin{enumerate}
		\item
		It dose not require all-to-all communication among users, which is often not feasible in many scenarios.
		\item
		It requires very low communication cost per user and at the server (See Table~\ref{table} for comparison).
		\item
	It is resilient to user dropouts. 
		\item
		It achieves worst-case information-theoretic security against a curious server and any subset of $T < N-D$ colluding users.
	\end{enumerate}
	
In addition, \tabref{table} shows the comparison between different frameworks in secure aggregation problem in terms of the communication loads. For a fair comparison, we consider two metrics: server communication and per user communication. Server communication indicates the total size of all messages which are sent or received by the server, and per user communication denotes the total size of all messages which are sent by each user.
	\begin{table}
		\centering
		\caption{
		Communication loads of secure aggregation frameworks in federated learning.
		}\label{table}
		\resizebox{0.5\columnwidth}{!}{
			\begin{tabular}{||c |c c ||} 
				\hline
				Approach & Server comm.  & Per user comm.   \\ [0.5ex] 
				\hline\hline
				\google~\cite{bonawitz2016practical} & $\mathcal{O}(NL+N^2)$ & $\mathcal{O}(L+N)$  \\
				\hline
				\googlep~\cite{bell2020secure}& $\mathcal{O}(NL+N\log N)$ & $\mathcal{O}(L+\log N)$  \\
				\hline
				
				\turbo~\cite{so2021turbo} & $\mathcal{O}(NL\log N)$ & $\mathcal{O}(L\log N)$ \\ 
				\hline
				Choi et al.\cite{choi2020communication}& $\mathcal{O}(N(\sqrt{N\log N}+L))$ & $\mathcal{O}(\sqrt{N\log N}+L)$ \\
				\hline
				\light~\cite{yang2021lightsecagg}& $\mathcal{O}(NL)$ & $\mathcal{O}(L)$\\
				\hline
				Proposed \swift & $(T+1)L$ & $(T+D+1)L$ \\ 
				\hline
			\end{tabular}
		}
	\end{table}

	
	The rest of the paper is organized as follows. In Section II, we formally formulate the problem in the proposed scheme. In Section III, we state the main result. In Section IV, we present the proposed scheme using a motivating example and the general form. Finally, we present the detailed proofs for the correctness of the proposed scheme and the privacy.
	
	\noindent {\bf Notation:} Matrices are denoted by upper boldface letters. For $n\in\mathbb{N}$ the notation $[n]$ represents set $\{1,\dots,n\}$. 
	Furthermore, the cardinality of set $\mathcal{S}$ is denoted by $|\mathcal{S}|$. In addition, we denote the difference of two sets $\mathcal{A}$, $\mathcal{B}$ as $\mathcal{A}\backslash\mathcal{B}$, that means the set of elements which belong to $\mathcal{A}$ but not $\mathcal{B}$. $H(X)$ denotes the entropy of random variable $X$ and $I(X;Y)$ is the mutual information of two random variables $X$ and $Y$.

\section{Problem formulation}
	We consider the secure aggregation problem, for a federated learning system consisting of a server and $N$ users ${U}_1,\dots,{U}_N$. For each $n \in [N]$, user $n$ has a  private local model of length $L$, denoted by  $\mathbf{W}_n\in\mathbb{F}^{L}$, for some finite field $\mathbb{F}$. Each user $n$ also has a collection of random variables $\mathcal{Z}_n$, whose elements are selected uniformly at random from $\mathbb{F}^L$, and independently from each other and from the local models. 
	Users can send messages to each other and also to the server, using error-free private communication links. 
	$\mathbf{M}^{(L)}_{n\to n'} \in \mathbb{F}^* \cup \{ \perp \}$ denotes  the message that user $n$ sends to user $n'$.  In addition, $\mathbf{X}^{(L)}_n \in \mathbb{F}^* \cup \{ \perp \}$ denotes the message sent by node $n$ to the server.  The null symbol $\perp$ represents the case no message is sent.

	The message $\mathbf{M}^{(L)}_{n\to n'}$ is a function 
	of $\mathbf{W}_n$, $\mathcal{Z}_n$, and the messages that node $n$ has received from other nodes so far. We denote the corresponding encoding function by $f^{(L)}_{n\to n'}$. 
	Similarly, $\mathbf{X}^{(L)}_n$ is a function of $\mathbf{W}_n$, $\mathcal{Z}_n$,  and the messages that node $n$ has received from other nodes so far. We denote the corresponding encoding function by $g^{(L)}_{n}$. For a subset ${\cal S} \subseteq [N]$, we let  $\mathcal{X}_{\mathcal{S}}=\{\mathbf{X}_n^{(L)}\}_{n\in\mathcal{S}}$ represent the set of messages the server receives from users in $\mathcal{S}$.
	We assume that a subset $\mathcal{D}\subset [N]$ of users drop out, i.e.,
	stay silent (or send $\perp$ to other nodes and the server) during the protocol execution. We denote the number of dropped out users as $D=|\mathcal{D}|$.
	
	We also assume that a subset
	$\mathcal{T}\subset[N]$ of the users, whose identities are not known before the execution of the aggregation protocol, are semi-honest. It means that users in $\mathcal{T}$ follow the protocol faithfully; however, they are curious and may collude with each other and with the server to gain information about the local models of the honest users. We assume $|\mathcal{T}|\le T$, for some security parameter $T < N-D$. 
	
	A secure aggregation scheme consists of the encoding functions  $f^{(L)}_{n\to n'}$ and $g^{(L)}_{n}$, $n,n' \in [N]$, such that the following conditions are satisfied: 
	
	
	\textbf{1. Correctness:} The server is  able to recover $\mathbf{W}~=~\sum_{n\in[N]\backslash\mathcal{D}}{\mathbf{W}}_{n}$,  using 
	$\mathcal{X}_{[N]\backslash\mathcal{D}}~=~\{\mathbf{X}^{(L)}_n\}_{n \in [N]\backslash\mathcal{D}}$. More
	precisely,
	\begin{align}
	H\bigg(\sum\limits_{n\in[N]\backslash\mathcal{D}}\mathbf{W}_n\big| \mathcal{X}_{[N]\backslash\mathcal{D}}\bigg)=0.
	\end{align}
	

	\textbf{2. Privacy Constraint:} Receiving $\mathcal{X}_{[N]\backslash\mathcal{D}}$, the server should not gain any information about local models of the honest users, beyond the aggregation of the local models, even if it colludes with semi-honest users in $\mathcal{T}$. Formally, 
	\begin{align}\label{serve-prv}
	\nonumber
	I\bigg(\mathbf{W}_n, {n\in[N]\backslash\mathcal{T}}; \mathcal{X}_{[N]\backslash\mathcal{D}},\bigcup\limits_{k\in\mathcal{T}}\{\mathbf{M}^{(L)}_{k'\to k}, k'\in[N]\},\{\mathbf{W}_k,\mathcal{Z}_k,{k\in\mathcal{T}}\}\bigg| \sum\limits_{n\in[N]\backslash\{\mathcal{D}\cup\mathcal{T}\}}{\mathbf{W}_n}\bigg)=0.
	\end{align}
	It is also possible for \swift to guarantee privacy when there are $T$ semi-honest users who collude with each other, and the server is curious, but does not collude with the semi-honest users. The privacy constraint here is that the server should not gain any information beyond the aggregation, and the semi-honest users should also not gain any information about the local models. 
	
	For a secure aggregation scheme satisfying the above two conditions, we define user-to-user communication load and uplink communication load as follows: 
	
	\begin{definition}[Normalized user-to-user communication load] denoted by $R^{(L)}_{\text{user}}$,   is defined as the 
		the aggregated size of all messages communicated between users, normalized by $L$, i.e.,
		\begin{align*}
		  R^{(L)}_{\text{user}} =\frac{1}{L} \sum_{n,n'\in [N]} H(\mathbf{M}^{(L)}_{n\to n'}).  
		\end{align*}
	\end{definition}
	\begin{definition}[Normalized uplink communication load] denoted by $R^{(L)}_{\text{uplink}}$,  is defined as the the aggregated size of all messages sent from users to the server, normalized by $L$, i.e.,
	\begin{align*}
	    R^{(L)}_{\text{uplink}}=\frac{1}{L} \sum_{n\in [N]} H(\mathbf{X}^{(L)}_{n}).
	\end{align*}
	\end{definition}
	We say that the pair of $(R_{\text{uplink}}, R_{\text{user}})$
	is achievable, if there exist a sequence of secure aggregation schemes with rate tuples $(R^{(L)}_{\text{uplink}}, R^{(L)}_{\text{user}})$, $L=1,2,\ldots$, such that 
	\begin{align}
	    R_{\text{uplink}}&=\limsup_{L \rightarrow \infty} R^{(L)}_{\text{uplink}},\\
	    R_{\text{user}}&=\limsup_{L \rightarrow \infty}R^{(L)}_{\text{user}}.
	\end{align}
	The capacity region of a secure aggregation problem, denoted by $\mathcal{C}_{N,D,T}$, is defined as the convex closure of all achievable rate tuples $(R_{\text{uplink}}, R_{\text{user}})$. 
	
	\section{Main result}
	We first present the main result of the proposed secure model aggregation scheme in the following theorem. 
	
	\begin{theorem}
		Consider a secure aggregation problem, with $N$ users and one server, where up to $T$ users are semi-honest and up to $D$ users may drop out. Let 
		\begin{align}\label{theorem}\nonumber
		& \mathcal{R}=\big\{  (R_{\text{uplink}} ,R_{\text{user}})\ | 
		R_{\text{uplink}} \geq (T+1),
		R_{\text{user}}\geq (N-1)(T+D+1)\big\},
		\end{align}  
		then $\mathcal{R} \subset \mathcal{C}_{N,D,T}$.
	\end{theorem}
	
	\begin{proof}
		The proof can be found in Section \ref{proof_comm}.
	\end{proof}
	
	To achieve the communication loads in~\eqref{theorem}, we propose \swift, a novel secure aggregation scheme, which partitions the users into disjoint groups and operates in two main phases: (i) Intra-group secret sharing and aggregation; and (ii) Inter-group communication and aggregation. Finally, communication with the server is required so that the server can obtain the aggregation of local models. 
	
	
	
	Compared to the existing schemes in secure aggregation, \swift simultaneously reduces both the uplink and the user-to-user communication cost (Table~\ref{table}), while the correctness and the privacy constraint are satisfied via information-theoretic approaches.

	\section{The Proposed Scheme}\label{sec2}
	In this section,  we propose \swift which reduces the communication load of the secure aggregation problem in federated learning.  We first introduce the main idea of this method using a simple motivation example. 
	
	\subsection{Motivation Example}
	Consider a secure aggregated problem consisting of one server and $N=12$ users, $U_1, U_2, \dots, U_{12}$. There is $D=1$ user dropout and up to $T=2$ semi-honest users that may collude with each other to gain some information about the local models of other users. User $n$ contains its local model $\mathbf{W}_n$ which is a vector with a length of $L$, $n\in[12]$.  Also, each user has two random vectors, $\mathcal{Z}_n=\{\mathbf{Z}_{n,1},\mathbf{Z}_{n,2} \}$ which are chosen uniformly at random from $\mathbb{F}^L$. Each user takes the following steps:
	\begin{enumerate}
		\item {\bf Grouping:}
		The set of users are arbitrarily partitioned into $\Gamma=3$ groups with a size of $ D+T+1=4$, denoted by $\mathcal{G}_1, \mathcal{G}_2, \mathcal{G}_3$.  Figure \ref{basic_scheme}  represents one example of this partitioning, where $\mathcal{G}_1=\{U_1, U_2, U_3, U_4\}$, $\mathcal{G}_2=\{U_5, U_6, U_7, U_8\}$, and $\mathcal{G}_3=\{U_9, U_{10}, U_{11}, U_{12}\}$. We also order the users in each group arbitrarily. For simplicity of exposition, we may refer to user $n$ based on its location in a group of users. If user $n$ is the $t$th user in group $\gamma$, we call it as user $(\gamma, t)$.  For example in \figref{basic_scheme} , user 8 is the same as user $(2,4)$/
		We use indices $n$ or $(\gamma, t)$ interchangeably.

		\item {\bf Intra-Group Secret Sharing and Aggregation:}
		User $n$ forms the following polynomial.
		\begin{align}
		\mathbf{F}_{n}(x)=\mathbf{W}_{n}+\mathbf{Z}_{n,1}x+\mathbf{Z}_{n,2}x^2,
		\end{align}
		where  $\mathbf{F}_{n}(0)=\mathbf{W}_{n}$ is the local model of  user $n$, $n\in[12]$. 
		
		Let $\alpha_t\in \mathbb{F}$, $t\in [4]$, are four distinct constants. We assign   $\alpha_t$ to user $t$ of all groups, i.e., users $(\gamma,t)$, $\gamma=1,\ldots, \Gamma$. 
		
		In this step, each user $(\gamma,t)$  sends the evaluation of its polynomial function at $\alpha_{t'}$, i.e.,  $\mathbf{F}_{(\gamma,t)}(\alpha_{t'})$, to user $(\gamma,t')$, for $t'\in [4]$. For example, in \figref{basic_scheme}, user $(2,1 )$, which is indeed $U_5$, sends
		$\mathbf{F}_{(2,1)}(\alpha_{1})= \mathbf{F}_{5}(\alpha_{1})$, 
		$\mathbf{F}_{(2,1)}(\alpha_{2})= \mathbf{F}_{5}(\alpha_{2})$, 
		$\mathbf{F}_{(2,1)}(\alpha_{3})= \mathbf{F}_{5}(\alpha_{3})$,
		$\mathbf{F}_{(2,1)}(\alpha_{4})= \mathbf{F}_{5}(\alpha_{4})$, to user $(2,1)$ (or user $U_5$ which is basically itself), user $(2,2)$ (or user $U_6$), user $(2,3)$ (user $U_7$), and user $(2,4)$ (user $U_8$) respectively. 
		If a user $(\gamma,t)$ drops out and stays silent, $\mathbf{F}_{(\gamma,t)}(\alpha_{t'})$ is just presumed to be zero.
		
		Each user $(\gamma,t)$ calculates 
		\begin{align*}
		\mathbf{Q}_{(\gamma,t)}=\mathbf{F}_{(\gamma, 1)}(\alpha_{t})+\mathbf{F}_{(\gamma, 2)}(\alpha_{t})
		+\mathbf{F}_{(\gamma, 3)}(\alpha_{t})+\mathbf{F}_{(\gamma, 4)}(\alpha_{t}).
		\end{align*}
		
		In this example, assume that $U_7$ or user $(2,3)$ drops out and does not send its share to other users in the second group. Other users within the group treat its share as zero. In this phase, within each group, at most 12 communication take place. 
		\item  {\bf Inter-group Communication and Aggregation:} In this phase, user $(1,t)$, $t\in [4]$, calculates the following message.
		\begin{align}
		\mathbf{S}_{(1,t)} & = \mathbf{Q}_{(1,t)},   
		\end{align}
		and sends $\mathbf{S}_{(1,t)}$ to user $(2,t)$. 
		
		User $(2,t)$, $t\in[4]$,
		calculates  $\mathbf{S}_{(2,t)}$ as 
		\begin{align}
		\mathbf{S}_{(2,t)}  = \mathbf{S}_{(1,t)} +  \mathbf{Q}_{(2,t)},
		\end{align}
		upon receiving $\mathbf{S}_{(1,t)}$ and sends it to user $(3,t)$. If user $(2,t)$ does not receive $\mathbf{S}_{(1,t)}$, it also remains silent for the rest of the protocol. In this particular example that user 7 drops out, it sends no message to user 11, and thus user 11 also remains silent.

		\item 	{\bf Communication with the Server:}
		User $t$ of the last group, i.e., user $(3, t)$ calculates 
		\begin{align}
		\mathbf{S}_{(3,t)}  = \mathbf{S}_{(2,t)} +  \mathbf{Q}_{(3,t)},
		\end{align}
		and sends $\mathbf{S}_{(3,t)}$ to the server, for $t\in [4]$.  
		Clearly in this example, user 11 remains silent and sends nothing (or null message $\perp$) to the server. 
		
		\item {\bf Recovering the result:}
		Let us define 
		\begin{align*}
		\mathbf{F}(x)\triangleq\sum_{\substack{n=1 \\ n\neq 7}}^{12}\mathbf{F}_{n}(x)=\sum_{\substack{n=1 \\ n\neq 7}}^{12}\mathbf{W}_{n}+x\sum_{\substack{n=1 \\ n\neq 7}}^{12}\mathbf{Z}_{n,1}+x^2\sum_{\substack{n=1 \\ n\neq 7}}^{12}\mathbf{Z}_{n,2}.
		\end{align*}
		One can verify that  $\mathbf{S}_{(3,t)}$, for $t=1,2,4$ that are received by the server are indeed equal to $\mathbf{F}(\alpha_1)$, $\mathbf{F}(\alpha_2)$, $\mathbf{F}(\alpha_4)$.
		
		Since $\mathbf{F}(x)$ is a polynomial function of degree 2, based on Lagrange interpolation rule the server can recover all the coefficients of this polynomial. In particular, the server can recover $\mathbf{F}(0)=\sum_{\substack{n=1 \\ n\neq 7}}^{12}\mathbf{W}_{n}$. Thus, the server is able to recover the aggregation of local models of surviving users and the correctness constraint is satisfied.
	\end{enumerate} 
The privacy constraint will be proven formally later in Subsection \ref{basic_privacy} for the general case.
	\begin{figure}

		\centering

		\tikzset{every picture/.style={line width=0.75pt}} 
		\scalebox{1}{

\tikzset{every picture/.style={line width=0.75pt}} 

\begin{tikzpicture}[x=0.75pt,y=0.75pt,yscale=-1,xscale=1]

\draw  [color={rgb, 255:red, 65; green, 117; blue, 5 }  ,draw opacity=1 ] (50.73,43.73) -- (98.85,79.29) -- (98.73,155.3) -- (50.49,190.71) -- cycle ;
\draw [color={rgb, 255:red, 65; green, 117; blue, 5 }  ,draw opacity=1 ]   (90.75,88.77) -- (55.35,179.82) ;
\draw [color={rgb, 255:red, 65; green, 117; blue, 5 }  ,draw opacity=1 ]   (56.72,54.09) -- (92.4,146.08) ;
\draw  [color={rgb, 255:red, 65; green, 117; blue, 5 }  ,draw opacity=1 ] (47.96,100.73) -- (50.42,96.61) -- (53.13,100.57) ;
\draw  [color={rgb, 255:red, 65; green, 117; blue, 5 }  ,draw opacity=1 ] (69.14,92.2) -- (68.59,85.85) -- (73.44,89.99) ;
\draw  [color={rgb, 255:red, 65; green, 117; blue, 5 }  ,draw opacity=1 ] (96.26,110.16) -- (98.8,105.63) -- (101.62,109.99) ;
\draw  [color={rgb, 255:red, 65; green, 117; blue, 5 }  ,draw opacity=1 ] (82.8,102.78) -- (86.7,99.04) -- (87.46,104.39) ;
\draw  [color={rgb, 255:red, 65; green, 117; blue, 5 }  ,draw opacity=1 ] (70.48,61.46) -- (68.15,56.82) -- (73.34,56.92) ;
\draw  [color={rgb, 255:red, 65; green, 117; blue, 5 }  ,draw opacity=1 ] (79.3,167.14) -- (84.37,166.02) -- (83,171.03) ;
\draw  [color={rgb, 255:red, 65; green, 117; blue, 5 }  ,draw opacity=1 ] (101.81,129.23) -- (99.14,133.69) -- (96.44,129.25) ;
\draw  [color={rgb, 255:red, 65; green, 117; blue, 5 }  ,draw opacity=1 ] (53.59,139.95) -- (50.7,144.26) -- (48.23,139.7) ;
\draw  [color={rgb, 255:red, 65; green, 117; blue, 5 }  ,draw opacity=1 ] (72.67,144.16) -- (67.96,146.34) -- (68.22,141.15) ;
\draw  [color={rgb, 255:red, 65; green, 117; blue, 5 }  ,draw opacity=1 ] (87.51,126.56) -- (87.38,131.75) -- (82.85,129.22) ;
\draw  [color={rgb, 255:red, 65; green, 117; blue, 5 }  ,draw opacity=1 ] (80.77,62.9) -- (83.14,67.52) -- (77.95,67.47) ;
\draw  [fill={rgb, 255:red, 255; green, 255; blue, 255 }  ,fill opacity=1 ] (41.08,44.28) .. controls (41.08,37.56) and (46.52,32.12) .. (53.24,32.12) .. controls (59.96,32.12) and (65.4,37.56) .. (65.4,44.28) .. controls (65.4,50.99) and (59.96,56.44) .. (53.24,56.44) .. controls (46.52,56.44) and (41.08,50.99) .. (41.08,44.28) -- cycle ;
\draw  [color={rgb, 255:red, 0; green, 0; blue, 0 }  ,draw opacity=1 ][fill={rgb, 255:red, 255; green, 255; blue, 255 }  ,fill opacity=1 ] (88.56,81.6) .. controls (88.56,74.88) and (94.01,69.44) .. (100.72,69.44) .. controls (107.44,69.44) and (112.88,74.88) .. (112.88,81.6) .. controls (112.88,88.32) and (107.44,93.76) .. (100.72,93.76) .. controls (94.01,93.76) and (88.56,88.32) .. (88.56,81.6) -- cycle ;
\draw  [color={rgb, 255:red, 0; green, 0; blue, 0 }  ,draw opacity=1 ][fill={rgb, 255:red, 255; green, 255; blue, 255 }  ,fill opacity=1 ] (86.27,157.57) .. controls (86.27,150.86) and (91.71,145.41) .. (98.43,145.41) .. controls (105.15,145.41) and (110.59,150.86) .. (110.59,157.57) .. controls (110.59,164.29) and (105.15,169.73) .. (98.43,169.73) .. controls (91.71,169.73) and (86.27,164.29) .. (86.27,157.57) -- cycle ;
\draw  [fill={rgb, 255:red, 255; green, 255; blue, 255 }  ,fill opacity=1 ] (37.04,191.59) .. controls (37.04,184.87) and (42.49,179.43) .. (49.2,179.43) .. controls (55.92,179.43) and (61.36,184.87) .. (61.36,191.59) .. controls (61.36,198.3) and (55.92,203.75) .. (49.2,203.75) .. controls (42.49,203.75) and (37.04,198.3) .. (37.04,191.59) -- cycle ;
\draw  [color={rgb, 255:red, 65; green, 117; blue, 5 }  ,draw opacity=1 ] (181.53,44.24) -- (229.65,79.81) -- (229.53,155.81) -- (181.29,191.22) -- cycle ;
\draw [color={rgb, 255:red, 65; green, 117; blue, 5 }  ,draw opacity=1 ]   (221.55,89.28) -- (186.15,180.33) ;
\draw [color={rgb, 255:red, 65; green, 117; blue, 5 }  ,draw opacity=1 ]   (187.52,54.6) -- (223.2,146.59) ;
\draw  [color={rgb, 255:red, 65; green, 117; blue, 5 }  ,draw opacity=1 ] (178.76,101.25) -- (181.22,97.12) -- (183.93,101.09) ;
\draw  [color={rgb, 255:red, 65; green, 117; blue, 5 }  ,draw opacity=1 ] (213.6,103.29) -- (217.5,99.55) -- (218.26,104.91) ;
\draw  [color={rgb, 255:red, 65; green, 117; blue, 5 }  ,draw opacity=1 ] (201.54,61.84) -- (199.21,57.2) -- (204.41,57.3) ;
\draw  [color={rgb, 255:red, 65; green, 117; blue, 5 }  ,draw opacity=1 ] (209.7,167.25) -- (214.77,166.13) -- (213.4,171.14) ;
\draw  [color={rgb, 255:red, 65; green, 117; blue, 5 }  ,draw opacity=1 ] (232.21,129.35) -- (229.54,133.8) -- (226.84,129.37) ;
\draw  [color={rgb, 255:red, 65; green, 117; blue, 5 }  ,draw opacity=1 ] (184.39,139.66) -- (181.5,143.97) -- (179.03,139.41) ;
\draw  [color={rgb, 255:red, 65; green, 117; blue, 5 }  ,draw opacity=1 ] (204.13,143.27) -- (199.42,145.45) -- (199.69,140.26) ;
\draw  [color={rgb, 255:red, 65; green, 117; blue, 5 }  ,draw opacity=1 ] (218.31,126.67) -- (218.18,131.86) -- (213.65,129.33) ;
\draw  [color={rgb, 255:red, 65; green, 117; blue, 5 }  ,draw opacity=1 ] (211.84,62.61) -- (214.21,67.23) -- (209.02,67.18) ;
\draw  [fill={rgb, 255:red, 255; green, 255; blue, 255 }  ,fill opacity=1 ] (171.48,43.99) .. controls (171.48,37.27) and (176.92,31.83) .. (183.64,31.83) .. controls (190.36,31.83) and (195.8,37.27) .. (195.8,43.99) .. controls (195.8,50.7) and (190.36,56.15) .. (183.64,56.15) .. controls (176.92,56.15) and (171.48,50.7) .. (171.48,43.99) -- cycle ;
\draw  [fill={rgb, 255:red, 255; green, 255; blue, 255 }  ,fill opacity=1 ] (218.96,81.31) .. controls (218.96,74.6) and (224.41,69.15) .. (231.12,69.15) .. controls (237.84,69.15) and (243.28,74.6) .. (243.28,81.31) .. controls (243.28,88.03) and (237.84,93.47) .. (231.12,93.47) .. controls (224.41,93.47) and (218.96,88.03) .. (218.96,81.31) -- cycle ;
\draw  [color={rgb, 255:red, 0; green, 0; blue, 0 }  ,draw opacity=1 ][fill={rgb, 255:red, 255; green, 255; blue, 255 }  ,fill opacity=1 ] (216.67,157.28) .. controls (216.67,150.57) and (222.11,145.12) .. (228.83,145.12) .. controls (235.55,145.12) and (240.99,150.57) .. (240.99,157.28) .. controls (240.99,164) and (235.55,169.45) .. (228.83,169.45) .. controls (222.11,169.45) and (216.67,164) .. (216.67,157.28) -- cycle ;
\draw  [fill={rgb, 255:red, 255; green, 255; blue, 255 }  ,fill opacity=1 ] (167.44,191.3) .. controls (167.44,184.58) and (172.89,179.14) .. (179.6,179.14) .. controls (186.32,179.14) and (191.76,184.58) .. (191.76,191.3) .. controls (191.76,198.02) and (186.32,203.46) .. (179.6,203.46) .. controls (172.89,203.46) and (167.44,198.02) .. (167.44,191.3) -- cycle ;
\draw [color={rgb, 255:red, 74; green, 144; blue, 226 }  ,draw opacity=1 ]   (65.4,44.28) -- (171.48,43.99) ;
\draw [color={rgb, 255:red, 74; green, 144; blue, 226 }  ,draw opacity=1 ]   (112.88,81.6) -- (218.96,81.31) ;
\draw [color={rgb, 255:red, 74; green, 144; blue, 226 }  ,draw opacity=1 ]   (110.59,157.57) -- (216.67,157.28) ;
\draw [color={rgb, 255:red, 74; green, 144; blue, 226 }  ,draw opacity=1 ]   (61.36,191.59) -- (167.44,191.3) ;
\draw  [color={rgb, 255:red, 74; green, 144; blue, 226 }  ,draw opacity=1 ] (123.85,41.55) -- (128.25,44.31) -- (123.76,46.91) ;
\draw  [color={rgb, 255:red, 74; green, 144; blue, 226 }  ,draw opacity=1 ] (155.85,78.75) -- (160.25,81.51) -- (155.76,84.11) ;
\draw  [color={rgb, 255:red, 74; green, 144; blue, 226 }  ,draw opacity=1 ] (155.05,154.75) -- (159.45,157.51) -- (154.96,160.11) ;
\draw  [color={rgb, 255:red, 74; green, 144; blue, 226 }  ,draw opacity=1 ] (114.25,188.75) -- (118.65,191.51) -- (114.16,194.11) ;
\draw  [color={rgb, 255:red, 65; green, 117; blue, 5 }  ,draw opacity=1 ] (311.93,44.04) -- (360.05,79.61) -- (359.93,155.61) -- (311.69,191.02) -- cycle ;
\draw [color={rgb, 255:red, 65; green, 117; blue, 5 }  ,draw opacity=1 ]   (351.95,89.08) -- (316.55,180.13) ;
\draw [color={rgb, 255:red, 65; green, 117; blue, 5 }  ,draw opacity=1 ]   (317.92,54.4) -- (353.6,146.39) ;
\draw  [color={rgb, 255:red, 65; green, 117; blue, 5 }  ,draw opacity=1 ] (309.16,101.05) -- (311.62,96.92) -- (314.33,100.89) ;
\draw  [color={rgb, 255:red, 65; green, 117; blue, 5 }  ,draw opacity=1 ] (330.34,92.51) -- (329.79,86.16) -- (334.64,90.3) ;
\draw  [color={rgb, 255:red, 65; green, 117; blue, 5 }  ,draw opacity=1 ] (357.46,110.47) -- (360,105.94) -- (362.82,110.3) ;
\draw  [color={rgb, 255:red, 65; green, 117; blue, 5 }  ,draw opacity=1 ] (344,103.09) -- (347.9,99.35) -- (348.66,104.71) ;
\draw  [color={rgb, 255:red, 65; green, 117; blue, 5 }  ,draw opacity=1 ] (331.28,61.64) -- (328.95,57) -- (334.14,57.1) ;
\draw  [color={rgb, 255:red, 65; green, 117; blue, 5 }  ,draw opacity=1 ] (340.1,167.05) -- (345.17,165.93) -- (343.8,170.94) ;
\draw  [color={rgb, 255:red, 65; green, 117; blue, 5 }  ,draw opacity=1 ] (362.61,129.15) -- (359.94,133.6) -- (357.24,129.17) ;
\draw  [color={rgb, 255:red, 65; green, 117; blue, 5 }  ,draw opacity=1 ] (314.79,139.46) -- (311.9,143.77) -- (309.43,139.21) ;
\draw  [color={rgb, 255:red, 65; green, 117; blue, 5 }  ,draw opacity=1 ] (334.53,143.07) -- (329.82,145.25) -- (330.09,140.06) ;
\draw  [color={rgb, 255:red, 65; green, 117; blue, 5 }  ,draw opacity=1 ] (348.71,126.47) -- (348.58,131.66) -- (344.05,129.13) ;
\draw  [color={rgb, 255:red, 65; green, 117; blue, 5 }  ,draw opacity=1 ] (341.57,62.41) -- (343.94,67.03) -- (338.75,66.98) ;

\draw  [fill={rgb, 255:red, 255; green, 255; blue, 255 }  ,fill opacity=1 ] (301.88,43.79) .. controls (301.88,37.07) and (307.32,31.63) .. (314.04,31.63) .. controls (320.76,31.63) and (326.2,37.07) .. (326.2,43.79) .. controls (326.2,50.5) and (320.76,55.95) .. (314.04,55.95) .. controls (307.32,55.95) and (301.88,50.5) .. (301.88,43.79) -- cycle ;
\draw  [fill={rgb, 255:red, 255; green, 255; blue, 255 }  ,fill opacity=1 ] (349.36,81.11) .. controls (349.36,74.4) and (354.81,68.95) .. (361.52,68.95) .. controls (368.24,68.95) and (373.68,74.4) .. (373.68,81.11) .. controls (373.68,87.83) and (368.24,93.27) .. (361.52,93.27) .. controls (354.81,93.27) and (349.36,87.83) .. (349.36,81.11) -- cycle ;
\draw  [color={rgb, 255:red, 0; green, 0; blue, 0 }  ,draw opacity=1 ][fill={rgb, 255:red, 255; green, 255; blue, 255 }  ,fill opacity=1 ] (347.07,157.08) .. controls (347.07,150.37) and (352.51,144.92) .. (359.23,144.92) .. controls (365.95,144.92) and (371.39,150.37) .. (371.39,157.08) .. controls (371.39,163.8) and (365.95,169.25) .. (359.23,169.25) .. controls (352.51,169.25) and (347.07,163.8) .. (347.07,157.08) -- cycle ;
\draw  [fill={rgb, 255:red, 255; green, 255; blue, 255 }  ,fill opacity=1 ] (297.84,191.1) .. controls (297.84,184.38) and (303.29,178.94) .. (310,178.94) .. controls (316.72,178.94) and (322.16,184.38) .. (322.16,191.1) .. controls (322.16,197.82) and (316.72,203.26) .. (310,203.26) .. controls (303.29,203.26) and (297.84,197.82) .. (297.84,191.1) -- cycle ;
\draw [color={rgb, 255:red, 74; green, 144; blue, 226 }  ,draw opacity=1 ]   (195.8,44.08) -- (301.88,43.79) ;
\draw [color={rgb, 255:red, 74; green, 144; blue, 226 }  ,draw opacity=1 ]   (243.28,81.4) -- (349.36,81.11) ;
\draw [color={rgb, 255:red, 74; green, 144; blue, 226 }  ,draw opacity=1 ] [dash pattern={on 4.5pt off 4.5pt}]  (240.99,157.37) -- (347.07,157.08) ;
\draw [color={rgb, 255:red, 74; green, 144; blue, 226 }  ,draw opacity=1 ]   (191.76,191.39) -- (297.84,191.1) ;
\draw  [color={rgb, 255:red, 74; green, 144; blue, 226 }  ,draw opacity=1 ] (254.25,41.35) -- (258.65,44.11) -- (254.16,46.71) ;
\draw  [color={rgb, 255:red, 74; green, 144; blue, 226 }  ,draw opacity=1 ] (286.25,78.55) -- (290.65,81.31) -- (286.16,83.91) ;
\draw  [color={rgb, 255:red, 74; green, 144; blue, 226 }  ,draw opacity=1 ] (285.45,154.55) -- (289.85,157.31) -- (285.36,159.91) ;
\draw  [color={rgb, 255:red, 74; green, 144; blue, 226 }  ,draw opacity=1 ] (244.65,188.55) -- (249.05,191.31) -- (244.56,193.91) ;
\draw   (423.69,44.44) .. controls (423.69,40.38) and (426.98,37.09) .. (431.04,37.09) -- (453.09,37.09) .. controls (457.15,37.09) and (460.44,40.38) .. (460.44,44.44) -- (460.44,196.94) .. controls (460.44,201) and (457.15,204.29) .. (453.09,204.29) -- (431.04,204.29) .. controls (426.98,204.29) and (423.69,201) .. (423.69,196.94) -- cycle ;
\draw [color={rgb, 255:red, 74; green, 144; blue, 226 }  ,draw opacity=1 ]   (326.2,43.79) -- (423.69,44.44) ;
\draw  [color={rgb, 255:red, 74; green, 144; blue, 226 }  ,draw opacity=1 ] (387.05,41.35) -- (391.45,44.11) -- (386.96,46.71) ;
\draw [color={rgb, 255:red, 74; green, 144; blue, 226 }  ,draw opacity=1 ]   (373.68,81.11) -- (424.44,80.69) ;
\draw  [color={rgb, 255:red, 74; green, 144; blue, 226 }  ,draw opacity=1 ] (386.25,78.15) -- (390.65,80.91) -- (386.16,83.51) ;
\draw [color={rgb, 255:red, 74; green, 144; blue, 226 }  ,draw opacity=1 ] [dash pattern={on 4.5pt off 4.5pt}]  (371.39,157.08) -- (424.04,157.49) ;
\draw  [color={rgb, 255:red, 74; green, 144; blue, 226 }  ,draw opacity=1 ] (386.25,154.15) -- (390.65,156.91) -- (386.16,159.51) ;
\draw [color={rgb, 255:red, 74; green, 144; blue, 226 }  ,draw opacity=1 ]   (322.16,191.1) -- (422.44,190.69) ;
\draw  [color={rgb, 255:red, 74; green, 144; blue, 226 }  ,draw opacity=1 ] (386.65,188.15) -- (391.05,190.91) -- (386.56,193.51) ;
\draw  [color={rgb, 255:red, 208; green, 2; blue, 27 }  ,draw opacity=1 ][line width=1.5]  (221.37,157.28) .. controls (221.37,152.7) and (224.71,148.99) .. (228.83,148.99) .. controls (232.95,148.99) and (236.29,152.7) .. (236.29,157.28) .. controls (236.29,161.87) and (232.95,165.58) .. (228.83,165.58) .. controls (224.71,165.58) and (221.37,161.87) .. (221.37,157.28) -- cycle ; \draw  [color={rgb, 255:red, 208; green, 2; blue, 27 }  ,draw opacity=1 ][line width=1.5]  (223.56,151.42) -- (234.1,163.15) ; \draw  [color={rgb, 255:red, 208; green, 2; blue, 27 }  ,draw opacity=1 ][line width=1.5]  (234.1,151.42) -- (223.56,163.15) ;
\draw  [color={rgb, 255:red, 65; green, 117; blue, 5 }  ,draw opacity=1 ] (74.21,176.43) -- (69.12,177.45) -- (70.58,172.47) ;
\draw  [color={rgb, 255:red, 65; green, 117; blue, 5 }  ,draw opacity=1 ] (335.01,177.23) -- (329.92,178.25) -- (331.38,173.27) ;

\draw (44.4,39.02) node [anchor=north west][inner sep=0.75pt]  [font=\scriptsize]  {$U_{1}$};
\draw (93.07,76.08) node [anchor=north west][inner sep=0.75pt]  [font=\scriptsize]  {$U_{2}$};
\draw (90.13,152.75) node [anchor=north west][inner sep=0.75pt]  [font=\scriptsize]  {$U_{3}$};
\draw (40.8,187.42) node [anchor=north west][inner sep=0.75pt]  [font=\scriptsize]  {$U_{4}$};
\draw (174.8,38.73) node [anchor=north west][inner sep=0.75pt]  [font=\scriptsize]  {$U_{5}$};
\draw (223.47,75.79) node [anchor=north west][inner sep=0.75pt]  [font=\scriptsize]  {$U_{6}$};
\draw (221.33,152.46) node [anchor=north west][inner sep=0.75pt]  [font=\scriptsize]  {$U_{7}$};
\draw (171.2,187.13) node [anchor=north west][inner sep=0.75pt]  [font=\scriptsize]  {$U_{8}$};
\draw (305.2,38.53) node [anchor=north west][inner sep=0.75pt]  [font=\scriptsize]  {$U_{9}$};
\draw (352.67,75.59) node [anchor=north west][inner sep=0.75pt]  [font=\scriptsize]  {$U_{10}$};
\draw (349.33,151.46) node [anchor=north west][inner sep=0.75pt]  [font=\scriptsize]  {$U_{11}$};
\draw (299.6,185.33) node [anchor=north west][inner sep=0.75pt]  [font=\scriptsize]  {$U_{12}$};
\draw (425.44,111.6) node [anchor=north west][inner sep=0.75pt]  [font=\scriptsize] [align=left] {\begin{minipage}[lt]{23.75pt}\setlength\topsep{0pt}
\begin{flushright}
Server
\end{flushright}

\end{minipage}};
\draw (51.2,206.75) node [anchor=north west][inner sep=0.75pt]   [align=left] {\begin{minipage}[lt]{39.57pt}\setlength\topsep{0pt}
\begin{flushright}
Group 1
\end{flushright}

\end{minipage}};
\draw (181.6,206.46) node [anchor=north west][inner sep=0.75pt]   [align=left] {\begin{minipage}[lt]{39.57pt}\setlength\topsep{0pt}
\begin{flushright}
Group 2
\end{flushright}

\end{minipage}};
\draw (312,206.26) node [anchor=north west][inner sep=0.75pt]   [align=left] {\begin{minipage}[lt]{39.57pt}\setlength\topsep{0pt}
\begin{flushright}
Group 3
\end{flushright}

\end{minipage}};

\end{tikzpicture}
			}
		\caption{An example of how users $U_1,U_2,\dots,U_{12}$ in \swift are partitioned into 3 different groups, where $T=2$ users are semi-hones and $D=1$ user may drop out. \swift consists of two main phases: (1) intra-group communication, shown by the green directed lines; (2) inter-group communication, shown by the blue directed lines. Dashed lines indicate no communication occurs in this direction.}
		\label{basic_scheme}
	\end{figure}
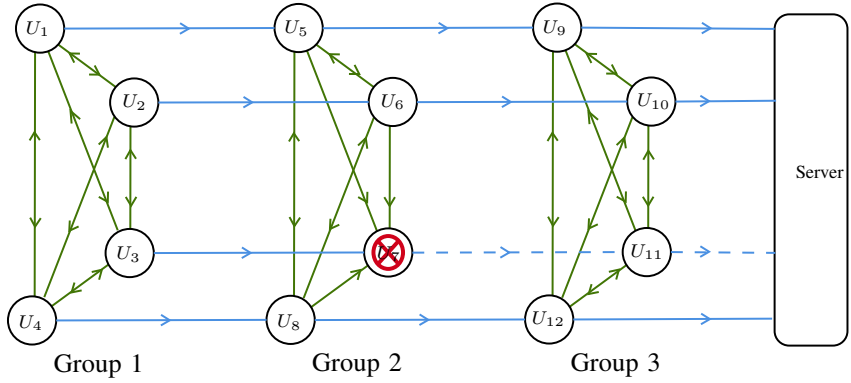
	\subsection{General case}
	In this subsection, we formally describe \swift. Consider a network consisting of one server and $N$ users, $U_1, U_2,\dots, U_N$, where  up to $T$ of them are semi-honest which may collude with each other to gain some information about other users. Furthermore, $D$ users may drop out, and their indices are denoted by $\mathcal{D}$.  The $n$th user contains its local model $\mathbf{W}_n\in \mathbb{F}^{L}$ and a set of random variables $\mathcal{Z}_n=\{\mathbf{Z}_{n,j},j\in[T] \}$ which are chosen independently and uniformly at random from $\mathbb{F}^{L}$. In this setting, the server here wants to recover the aggregated local models of the  surviving users, i.e., $\mathbf{W}=\sum_{n\in[N]\backslash\mathcal{D}}\mathbf{W}_{n}$, while the individual models remain private. To reach this goal, \swift takes the following steps.
	\begin{enumerate}
		\item {\bf Grouping:}
		The set of $N$ users are arbitrarily partitioned into $\Gamma$ groups with a size of $ \nu\triangleq D+T+1$, denoted by $\mathcal{G}_1, \mathcal{G}_2, \dots \mathcal{G}_{\Gamma}$, and it is divisible by $\nu$.  Figure \ref{fig3}  represents an overview of the proposed grouping method. We also order the users in the each group arbitrarily.  For simplicity, we refer to user $n$ based on its location in a group of users. If user $n$ is the $t$th user in group $\gamma$, we call it as user $(\gamma, t)$. 
		\item {\bf Intra-Group Secret Sharing and Aggregation:}
		User $n\in[N]$ forms the following polynomial.
		\begin{align}\label{Fn}
		\mathbf{F}_n(x)=\mathbf{W}_n+\sum\limits_{j=1}^{T}\mathbf{Z}_{n,j}x^j.
		\end{align}
		This polynomial function is designed such that $\mathbf{F}_n(0)=\mathbf{W}_n$. Each user uses its polynomial function $\mathbf{F}_n(.)$ to share its local model with other users.
		
		
		Let $\alpha_t\in \mathbb{F}$, $t\in [\nu]$, are $\nu$ distinct constants. We assign   $\alpha_t$ to user $t$ of all groups, i.e., users $(\gamma,t)$, $\gamma=1,\ldots, \Gamma$. 
		
		In this step, each user $(\gamma,t)$  sends the evaluation of its polynomial function at $\alpha_{t'}$, i.e.,  $\mathbf{F}_{(\gamma,t)}(\alpha_{t'})$, to user $(\gamma,t')$, for $t'\in [\nu]$. 
		If a user $(\gamma,t)$ drops out and stays silent, $\mathbf{F}_{(\gamma,t)}(\alpha_{t'})$ is just presumed to be zero.
		
		Each user $(\gamma,t)$ calculates 
		\begin{align}
		\mathbf{Q}_{(\gamma,t)} = \sum_{t' \in [\nu]}  \mathbf{F}_{(\gamma, t')}(\alpha_{t}).
		\end{align}
		
		Note that in this phase, within each group, at most $\nu(\nu-1)$ communication take place. 
		\item  {\bf Inter-group Communication and Aggregation:} In this phase, user $t$ of group $\gamma$ calculates a message denoted by $\mathbf{S}_{(\gamma,t)}$ and  sends it to user $t$ of group $\gamma+1$, for $\gamma=1,\ldots, \Gamma-1$. 
		
		User $(1,t)$, $t\in [\nu]$, in group one sets
		\begin{align}\label{recursive1}
		\mathbf{S}_{(1,t)} & = \mathbf{Q}_{(1,t)}   
		\end{align}
		and sends $\mathbf{S}_{(1,t)}$ to user $(2,t)$. 
		
		User $(\gamma,t)$ in group $\gamma$, $2\le\gamma \le \Gamma-1$,
		calculates  $\mathbf{S}_{(\gamma,t)}$ as 
		\begin{align}\label{recursive2}
		\mathbf{S}_{(\gamma,t)}  = \mathbf{S}_{(\gamma-1,t)} +  \mathbf{Q}_{(\gamma,t)},
		\end{align}
		upon receiving $\mathbf{S}_{(\gamma-1,t)}$. If user $(\gamma,t)$ does not receive $\mathbf{S}_{(\gamma-1,t)}$, it also remains silent for the rest of the protocol.
		\item 	{\bf Communication with the Server:}
		User $t$ of the last group, i.e., user $(\Gamma, t)$
		computes
		\begin{align}\label{recursive3}
		\mathbf{S}_{(\Gamma,t)}  = \mathbf{S}_{(\Gamma-1,t)} +  \mathbf{Q}_{(\Gamma,t)},
		\end{align}
		and sends it to the server, for $t\in [\nu]$.  
		
		\item {\bf Recovering the result:}
		Having received the outcomes of a subset of users in $\mathcal{G}_{\Gamma}$ with a size of at least $T+1$, the server can recover the aggregated local models.
	\end{enumerate}
	\begin{figure}
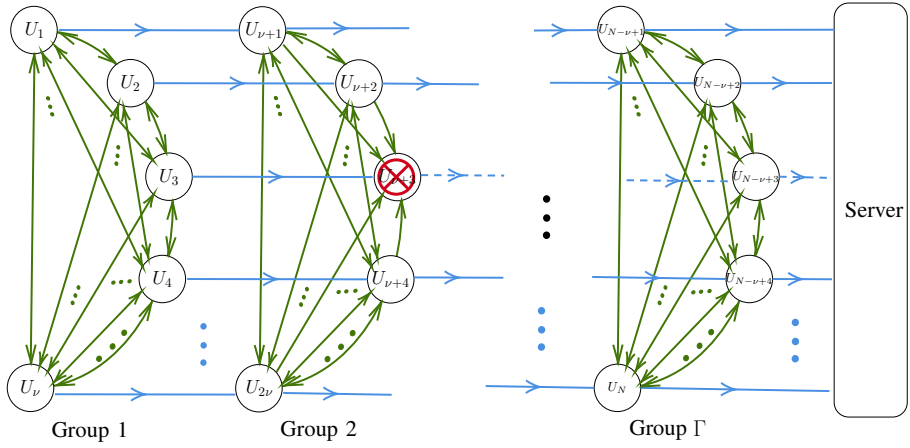

		\centering

		\tikzset{every picture/.style={line width=0.75pt}} 
		\scalebox{0.5}{

\tikzset{every picture/.style={line width=0.75pt}} 

	

		}
		\caption{An overview of the proposed setting in \swift, where $\nu=T+D+1$. The intra-group and inter-group communication links are shown in green and blue, respectively.}
		\label{fig3}
	\end{figure}
	\subsection{Proof of correctness}\label{proof_correct}
	
	To prove the correctness we must show that the server can recover $\sum_{n\in[N]\backslash\mathcal{D}}\mathbf{W}_{n}$ from the messages received from group $\Gamma$.  
	Using the recursive \cref{recursive1,recursive2,recursive3}, $\mathbf{S}_{(\gamma,t)}(\alpha_t)$ is either a null message, or it is equal to
	\begin{align}\label{S_gamma}
	\mathbf{S}_{(\gamma,t)}(\alpha_t)=\sum\limits_{\gamma'=1}^{\gamma}\sum\limits_{n\in \mathcal{G}_{\gamma'} \backslash \mathcal{D} }\hspace{-3mm}\mathbf{W}_n+\sum\limits_{j=1}^{T}{\alpha_t}^j\sum\limits_{\gamma'=1}^{\gamma}\sum\limits_{{n\in \mathcal{G}_{\gamma'} \backslash \mathcal{D} }
	}\mathbf{Z}_{n,j},
	\end{align}
	Thus if user $t$ in group $\Gamma$ sends a message to the server, it is equal to 	$\mathbf{S}_{(\Gamma,t)}(\alpha_t)$. From \eqref{S_gamma}, it is easy to see that $\mathbf{S}_{(\Gamma,t)}(\alpha_t)=\mathbf{F}(\alpha_t)$, where 
	\begin{align}\label{F}
	\mathbf{F}(x)=\sum\limits_{n\in[N]\backslash\mathcal{D}}\mathbf{W}_n+ \sum\limits_{j=1}^{T}x^j\sum\limits_{n\in[N]\backslash\mathcal{D}}\mathbf{Z}_{n,j}.
	\end{align}
	$\mathbf{F}(x)$ is a polynomial of degree $T$, with $\mathbf{F}(0)=\sum_{\substack{n\in[N]\backslash\mathcal{D}}}\mathbf{W}_{n}$.  Thus if the server receives at least $T+1$ messages from the last group, it can use Lagrange interpolation to recover  $\mathbf{F}(x)$ and $\sum_{n\in[N]\backslash\mathcal{D}}\mathbf{W}_{n}$.
	
	Recall that, in \swift, for any user $(\gamma,t)$ in $\mathcal{D}$, all the messages $\mathbf{S}_{(\gamma',t)}(\alpha_t)$, $\gamma \le \gamma' \le \Gamma$ is null. In particular, for any user $(\gamma,t)$ in $\mathcal{D}$,  $\mathbf{S}_{(\Gamma,t)}(\alpha_t)$ is null. Thus at most $D$ users in the last group send null messages to the server. Since the size of each group is $D+T+1$, the server receives at least $T+1$ values  $\mathbf{S}_{(\Gamma,t)}(\alpha_t)$ for distinct $\alpha_t$, and thus can recover $\mathbf{F}(x)$. 

	\subsection{ The communication loads}\label{proof_comm}
	According to \eqref{F}, the total number of messages that are needed to be received by the server is $(T+1)$. Thus, the normalized uplink communication load in \swift is $R_{\text{uplink}}^{(L)}=(T+1)$. In each group, at most $\nu(\nu-1)$ messages are sent by the members, and there are $\frac{N}{\nu}$ groups. In addition, at most $\nu$ messages are sent between two consecutive groups. Thus, the normalized user-to-user communication load of \swift is upper-bounded as $R_{\text{user}}^{(L)}\le(N-1)\nu$, where $\nu=T+D+1$.	
	\subsection{Proof of privacy}\label{basic_privacy}
In this section, we prove that \swift satisfies the privacy constraint in \eqref{serve-prv}. The privacy must be guaranteed even if the server colludes with any set $\mathcal{T}\subset [N]$ of at most $T$ semi-honest users which can distribute arbitrary across the groups. 
At a high level, we expand the mutual information in \eqref{serve-prv} over the groups containing the semi-honest users, from the first group to the last one, and show that model privacy will be preserved at each expansion step.
	\begin{corollary}\label{coro4}
		Assume that user $U_n$ is denoted by $(\gamma,t)$.  In \swift, the local model of $U_n$, is shared using polynomial function $\mathbf{F}_n(x)$ in \eqref{Fn}. In other words, $\mathbf{F}_{(\gamma,t)}(\alpha_{t'})$ for $t'\in[\nu]\backslash \{t\}$ are delivered to user $(\gamma,t')$.
		According to \eqref{Fn} and directly from the privacy guarantee in Shamir's sharing, we have
		$I(\mathbf{W}_{n};\{ \mathbf{F}_n(\alpha_{t'}).t'\in\mathcal{T}\})=0.$
	\end{corollary}
	Let the random part of inter-group message $\mathbf{S}_{(\gamma,t)}$ of user $(\gamma,t)$, consisting of the random noises of non-dropped and honest users in the groups $\gamma'\le\gamma$, be denoted by $\tilde{\mathbf{Z}}_{(\gamma,t)}$, i.e.,
	\begin{align}\label{z_tild}
	    \tilde{\mathbf{Z}}_{(\gamma,t)}\triangleq\sum_{j=1}^{T}{\alpha_t}^{j}\sum_{\gamma'\in[\gamma]}\sum_{{n\in \mathcal{G}_{\gamma'} \backslash \{\mathcal{D}\cup \mathcal{T}\} }}\mathbf{Z}_{n,j}.
	\end{align}
We show in the following lemma, that the random noise in the message sent from group $\gamma$ to group $\gamma+1$ is independent of the random noise in the message sent from group $\gamma+1$ to group $\gamma+2$. 
\begin{lemma}\label{lemma1}
For all $\gamma\in[\Gamma]$ and $t\in[\nu]$, we have 
   $ I\big(\tilde{\mathbf{Z}}_{(\gamma,t)};\tilde{\mathbf{Z}}_{(\gamma+1,t)}\big)=0.$
\end{lemma}
	
\begin{proof}
	     Let us define
	     \begin{align}\label{z'}
	         {\mathbf{Z}'}_{n}^{(t)}\triangleq\sum_{j=1}^{T}\mathbf{Z}_{n,j}\alpha_{t}^{j},
	     \end{align}
	     for $n\in[N]$ and $t\in[\nu]$. In each group, there are $\nu=T+D+1$ users each of which uses $T$ random vectors chosen uniformly and independently from $\mathbb{F}^{L}$ in its shares. In addition, we have
	   $ \tilde{\mathbf{Z}}_{(\gamma+1,t)}=\tilde{\mathbf{Z}}_{(\gamma,t)}+\sum_{i\in\mathcal{G}_{\gamma+1}\backslash\{\mathcal{D}\cup \mathcal{T}\}}{\mathbf{Z}'}_i^{(t)}$. Since $\mathcal{G}_{\gamma+1}\backslash\{\mathcal{D}\cup \mathcal{T}\} \neq \emptyset$, and 
	    the random vectors $\{{\mathbf{Z}'}_i^{(t)}\}_{i\in\mathcal{G}_{\gamma+1}\backslash\{\mathcal{D}\cup \mathcal{T}\}}$ are i.i.d., we have that $ I\big(\tilde{\mathbf{Z}}_{(\gamma,t)};\tilde{\mathbf{Z}}_{(\gamma+1,t)}\big)=0$ for $\gamma\in[\Gamma]$.
	\end{proof}
	Assume that the semi-honest users are denoted by $\tilde{U}_1, \tilde{U}_2,\dots,\tilde{U}_T$. We denote the indices of theses semi-honest users as $(\gamma_1,t_1),(\gamma_2,t_2),\dots,(\gamma_T,t_T)$ respectively, where $1\le\gamma_1\le\gamma_2\le\dots\le\gamma_T\le\Gamma$. 
	 We also denote the set of indices of honest users in group $\gamma$, $\gamma\in[\Gamma]$, by $\mathcal{H}_{\gamma}\triangleq\{n:U_n\in\mathcal{G}_{\gamma}\backslash\{\mathcal{T}\cup \mathcal{D}\}\}$.
	 
	     Let us define the set of messages which are received by $\tilde{U}_{i}$ by $\mathcal{M}_{\tilde{U}_{i}}$ which consists of two kinds of messages. Particularly, $\mathcal{M}_{\tilde{U}_{i}}=\big\{\{\mathbf{F}_{(n,t_i)},n\in\mathcal{H}_{\gamma_i}\},\mathbf{S}_{(\gamma_i-1,t_i)}  \big\}$,
	    where
 	    $ \{\mathbf{F}_{(n,t_i)},n\in\mathcal{H}_{\gamma_i}\}
 	    =\{\mathbf{W}_n+{\mathbf{Z}'}_{n}^{(t_i)}, n\in\mathcal{H}_{\gamma_i}\}$,
 	is a set of intra-group messages, and $ \mathbf{S}_{(\gamma_i-1,t_i)}=\sum_{\gamma'\in[\gamma_i-1]}\sum_{m\in \mathcal{H}_{\gamma'}} \mathbf{W}_m+\tilde{\mathbf{Z}}_{(\gamma_i-1,t_i)},$
 	  is the message received from group $\gamma_i-1$. 
 	  
 According to the Shamir's secret sharing scheme, definitions in \eqref{z_tild},\eqref{z'},  and the fact that the random vectors are chosen uniformly and independently at random from $\mathbb{F}^{L}$, we can easily prove the following lemmas.
 
\begin{lemma}\label{lemma2}
For each user $(\gamma_i,t_i)$ in \swift we have $I\big( {\tilde{\mathbf{Z}}}_{(\gamma_i-1,t_i)};\{{\mathbf{Z}'}_n^{(t_i)},n\in\mathcal{H}_{\gamma_i}\} \big)=0$, for all $t_i\in[\nu]$ and $\gamma_i\in[\Gamma]$.
\end{lemma}
	 \begin{lemma}\label{lemma3}
	 Consider user $(\gamma_i,t_i)$ and user $(\tilde{\gamma}_i,\tilde{t}_i)$. Then, $I\big( \{{\mathbf{Z}'}_n^{(t_i)},n\in\mathcal{H}_{\gamma_i}\}; \{{\mathbf{Z}'}_{{n}}^{(\tilde{t}_i)},{n}\in\mathcal{H}_{\tilde{\gamma}_i}\}\big)=0$ for $\gamma_i\neq \tilde{\gamma}_i$, $t_i,\tilde{t}_i\in[\nu]$.
	 \end{lemma} 
	 \begin{lemma}\label{lemma4}
 For any $\gamma_i\in[\Gamma]$ consider user $n\in\mathcal{H}_{\gamma_i}$. Then
 for all $t_i\in[\nu]$, $I\big( {\mathbf{Z}'}_n^{(t_i)};\{{\mathbf{Z}'}_n^{(t)},t\in\mathcal{T}'\}\big)=0$, where $\mathcal{T}'\subset[\nu]\backslash \{t_i\}$, and $|\mathcal{T}'|\le T-1$. Similarly, $I\big( \tilde{\mathbf{Z}}_{(\gamma_i,t_i)};\{\tilde{\mathbf{Z}}_{(\gamma_i,t)},t\in\mathcal{T}'\}\big)=0$.
\end{lemma}
Let us define $\mathcal{W}_{N\backslash\mathcal{T}}\triangleq\{\mathbf{W}_n,{n\in[N]}\backslash{\mathcal{T}}\}$, $\mathcal{K}_{N,\mathcal{T}}\triangleq\{ \{\mathbf{W}_k,\mathcal{Z}_k,{k\in\mathcal{T}}\},\sum_{n\in[N]\backslash\{\mathcal{D}\cup\mathcal{T}\}}{\mathbf{W}_n}\}$, and  $\mathcal{M}_{\mathcal{T}}~\triangleq~\bigcup_{i\in[T]}\mathcal{M}_{\tilde{U}_i}$. In addition, the set of messages that the server receives from users in group $\Gamma$ are represented by $\mathcal{S}_{\Gamma}~\triangleq~\{\mathbf{S}_{(\Gamma,t)},{t\in[\nu]\backslash\mathcal{D}}\}$.
From the definition of privacy constraint, we have 
	\begin{align}\nonumber
	I\big(&\mathcal{W}_{N\backslash\mathcal{T}};\mathcal{M}_{\mathcal{T}},\mathcal{S}_{\Gamma},\{\mathbf{W}_k,\mathcal{Z}_k,{k\in\mathcal{T}}\}\big|\sum\limits_{n\in[N]\backslash\{\mathcal{D}\cup\mathcal{T}\}}{\mathbf{W}_n} \big)\\\nonumber
	=&I\big(\mathcal{W}_{N\backslash\mathcal{T}};\{\mathbf{W}_k,\mathcal{Z}_k,{k\in\mathcal{T}}\}\big|\sum\limits_{n\in[N]\backslash\{\mathcal{D}\cup\mathcal{T}\}}{\mathbf{W}_n} \big)+I\big(\mathcal{W}_{N\backslash\mathcal{T}};\mathcal{M}_{\mathcal{T}},\mathcal{S}_{\Gamma}\big|\sum\limits_{n\in[N]\backslash\{\mathcal{D}\cup\mathcal{T}\}}{\mathbf{W}_n},\{\mathbf{W}_k,\mathcal{Z}_k,{k\in\mathcal{T}}\} \big)\\\label{25}
	\stackrel{\text{(a)}}{=}&I\big(\mathcal{W}_{N\backslash\mathcal{T}};\mathcal{M}_{\mathcal{T}},\mathcal{S}_{\Gamma}\big|\sum\limits_{n\in[N]\backslash\{\mathcal{D}\cup\mathcal{T}\}}{\mathbf{W}_n},\{\mathbf{W}_k,\mathcal{Z}_k,{k\in\mathcal{T}}\} \big),
		\end{align}
	where in $(a)$ we use the independence of the local models and independence of random vectors from the local models.

	\begin{lemma}\label{lemma7}
	Let $(\gamma_1,t_1),(\gamma_2,t_2),\dots,(\gamma_T,t_T)$ be $T$ semi-honest users, where $1\le\gamma_1\le\gamma_2\le\dots\le\gamma_T\le\Gamma$ and $t_i\in[\nu]$ for $i\in[T]$. Then, for $i\in[T]$ we have
	\begin{align*}
	  	I\big(\mathcal{W}_{N\backslash\mathcal{T}};\mathcal{M}_{\tilde{U}_i}\big|\mathcal{K}_{N,\mathcal{T}},\{\mathcal{M}_{\tilde{U}_j},j\in[i-1]\} \big)=0. 
	\end{align*}

\begin{proof}
	    For $i=1$, using the independence of the local models and random variables, we have 
 	    \begin{align*}
	        I\big(\mathcal{W}_{N\backslash\mathcal{T}};\mathcal{M}_{\tilde{U}_1}\big|\mathcal{K}_{N,\mathcal{T}} \big) =&I\big(\mathcal{W}_{N\backslash\mathcal{T}};\{\mathbf{F}_{(n,t_1)}, n\in\mathcal{H}_{\gamma_1}\},\mathbf{S}_{(\gamma_1-1,t_1)}\big|\mathcal{K}_{N,\mathcal{T}} \big)\\
	        =&H\big(\{\mathbf{F}_{(n,t_1)}, n\in\mathcal{H}_{\gamma_1}\},\mathbf{S}_{(\gamma_1-1,t_1)}\big|\mathcal{K}_{N,\mathcal{T}} \big)-H\big(\{\mathbf{F}_{(n,t_1)}, n\in\mathcal{H}_{\gamma_1}\},\mathbf{S}_{(\gamma_1-1,t_1)}\big|\mathcal{K}_{N,\mathcal{T}},\mathcal{W}_{N\backslash\mathcal{T}}\big)\\
	        \leq &H\big(\{\mathbf{F}_{(n,t_1)}, n\in\mathcal{H}_{\gamma_1}\},\mathbf{S}_{(\gamma_1-1,t_1)}\big)
	    -H\big(\{{\mathbf{Z}'}_n^{(t_1)},n\in\mathcal{H}_{\gamma_1}\},\tilde{\mathbf{Z}}_{(\gamma_1-1,t_1)} \big)\le 0.
	    \end{align*}
	    The last term follows from the fact that 
	    $H\big(\{\mathbf{F}_{(n,t_1)}, n\in\mathcal{H}_{\gamma_1}\},\mathbf{S}_{(\gamma_1-1,t_1)}\big)$ and
	    $H\big(\{{\mathbf{Z}'}_n^{(t_1)},n\in\mathcal{H}_{\gamma_1}\},\tilde{\mathbf{Z}}_{(\gamma_1-1,t_1)} \big)$
	    have the same size and uniform variables maximize the entropy.  Therefore, $I\big(\mathcal{W}_{N\backslash\mathcal{T}};\mathcal{M}_{\tilde{U}_1}\big|\mathcal{K}_{N,\mathcal{T}} \big)=0$.

For $i\in[2:T]$ we have
	    \begin{align}\nonumber
	        I\big(&\mathcal{W}_{N\backslash\mathcal{T}};\mathcal{M}_{\tilde{U}_i}\big|\mathcal{K}_{N,\mathcal{T}},\{\mathcal{M}_{\tilde{U}_j},j\in[i-1]\} \big)\\\nonumber
	        =&I\big(\mathcal{W}_{N\backslash\mathcal{T}};\{\mathbf{F}_{(n,t_i)}, n\in\mathcal{H}_{\gamma_i}\},\mathbf{S}_{(\gamma_i-1,t_i)}\big|\mathcal{K}_{N,\mathcal{T}},
	         \big\{\{\mathbf{F}_{(n,t_{i-\ell})}, n\in\mathcal{H}_{\gamma_{i-\ell}}\},\mathbf{S}_{(\gamma_{i-\ell}-1,t_{i-\ell})},\ell\in[i-1]\big\} \big)\\\nonumber
	        =&I\big(\mathcal{W}_{N\backslash\mathcal{T}};\{\mathbf{F}_{(n,t_i)}, n\in\mathcal{H}_{\gamma_i}\}\big|\mathcal{K}_{N,\mathcal{T}},
	         \big\{\{\mathbf{F}_{(n,t_{i-\ell})}, n\in\mathcal{H}_{\gamma_{i-\ell}}\},\mathbf{S}_{(\gamma_{i-\ell}-1,t_{i-\ell})},\ell\in[i-1]\big\} \big)\\\label{eq_I}
	        &+I\big(\mathcal{W}_{N\backslash\mathcal{T}};\mathbf{S}_{(\gamma_i-1,t_i)}\big|\mathcal{K}_{N,\mathcal{T}},\{\mathbf{F}_{(n,t_i)}, n\in\mathcal{H}_{\gamma_i}\},
	         \big\{\{\mathbf{F}_{(n,t_{i-\ell})}, n\in\mathcal{H}_{\gamma_{i-\ell}}\},\mathbf{S}_{(\gamma_{i-\ell}-1,t_{i-\ell})},\ell\in[i-1]\big\} \big)
	         \end{align}
	          Using the definition of mutual information, \eqref{eq_I} can be written as
	         \begin{align}
	        &H\big(\{\mathbf{F}_{(n,t_i)}, n\in\mathcal{H}_{\gamma_i}\}\big|\mathcal{K}_{N,\mathcal{T}},\big\{\{\mathbf{F}_{(n,t_{i-\ell})}, n\in\mathcal{H}_{\gamma_{i-\ell}}\},
	        \mathbf{S}_{(\gamma_{i-\ell}-1,t_{i-\ell})},\ell\in[i-1]\big\}\big)\\\nonumber
	        &- H\big(\{\mathbf{F}_{(n,t_i)}, n\in\mathcal{H}_{\gamma_i}\}\big|\mathcal{K}_{N,\mathcal{T}}, \big\{\{\mathbf{F}_{(n,t_{i-\ell})}, n\in\mathcal{H}_{\gamma_{i-\ell}}\},
	         \mathbf{S}_{(\gamma_{i-\ell}-1,t_{i-\ell})},\ell\in[i-1]\big\},\mathcal{W}_{N\backslash\mathcal{T}}\big)\\\nonumber
	        &+H\big(\mathbf{S}_{(\gamma_i-1,t_i)}\big|\mathcal{K}_{N,\mathcal{T}},\{\mathbf{F}_{(n,t_i)}, n\in\mathcal{H}_{\gamma_i}\},
	         \big\{\{\mathbf{F}_{(n,t_{i-\ell})}, n\in\mathcal{H}_{\gamma_{i-\ell}}\},\mathbf{S}_{(\gamma_{i-\ell}-1,t_{i-\ell})},\ell\in[i-1]\big\}\big)\\\nonumber
	        &-H\big(\mathbf{S}_{(\gamma_i-1,t_i)}\big|\mathcal{K}_{N,\mathcal{T}},\{\mathbf{F}_{(n,t_i)}, n\in\mathcal{H}_{\gamma_i}\},
	        \big\{\{\mathbf{F}_{(n,t_{i-\ell})}, n\in\mathcal{H}_{\gamma_{i-\ell}}\},\mathbf{S}_{(\gamma_{i-\ell}-1,t_{i-\ell})},\ell\in[i-1]\big\},\mathcal{W}_{N\backslash\mathcal{T}}\big)\\\label{eq18}
	        \stackrel{\text{(a)}}{\leq }&H\big(\{\mathbf{F}_{(n,t_i)}, n\in\mathcal{H}_{\gamma_i}\}\big)-H\big(\{{\mathbf{Z}'}_n^{(t_i)}, n\in\mathcal{H}_{\gamma_i}\}\big|\mathcal{K}_{N,\mathcal{T}}, \big\{\{{\mathbf{Z}'}_n^{(t_{i-\ell})}, n\in\mathcal{H}_{\gamma_{i-\ell}}\},
	         {\tilde{\mathbf{Z}}}_{(\gamma_{i-\ell}-1,t_{i-\ell})},\ell\in[i-1]\big\},\mathcal{W}_{N\backslash\mathcal{T}}\big)\\\nonumber
	       & +H\big(\mathbf{S}_{(\gamma_i-1,t_i)}\big)-H\big({\tilde{\mathbf{Z}}}_{(\gamma_i-1,t_i)}\big|\mathcal{K}_{N,\mathcal{T}},\{{\mathbf{Z}'}_n^{(t_i)}, n\in\mathcal{H}_{\gamma_i}\},
	        \big\{\{{\mathbf{Z}'}_n^{(t_{i-\ell})}, n\in\mathcal{H}_{\gamma_{i-\ell}}\},{\tilde{\mathbf{Z}}}_{(\gamma_{i-\ell}-1,t_{i-\ell})},\ell\in[i-1]\big\},\mathcal{W}_{N\backslash\mathcal{T}}\big),
	        \end{align}
	        where in (a) the first and the third terms follow from the fact that $H(X|Y)\le H(X)$. Now we show that
	        \begin{align}\label{eq19}
	            I\big(\{{\mathbf{Z}'}_n^{(t_i)}, n\in\mathcal{H}_{\gamma_i}\};\mathcal{K}_{N,\mathcal{T}}, \big\{\{{\mathbf{Z}'}_n^{(t_{i-\ell})}, n\in\mathcal{H}_{\gamma_{i-\ell}}\},
	         {\tilde{\mathbf{Z}}}_{(\gamma_{i-\ell}-1,t_{i-\ell})},\ell\in[i-1]\big\},\mathcal{W}_{N\backslash\mathcal{T}}\big)=0.
	        \end{align}
	        From definition of mutual information, we have
	        \begin{align*}
	           I\big(&\{{\mathbf{Z}'}_n^{(t_i)}, n\in\mathcal{H}_{\gamma_i}\};\mathcal{K}_{N,\mathcal{T}}, \big\{\{{\mathbf{Z}'}_n^{(t_{i-\ell})}, n\in\mathcal{H}_{\gamma_{i-\ell}}\},
	         {\tilde{\mathbf{Z}}}_{(\gamma_{i-\ell}-1,t_{i-\ell})},\ell\in[i-1]\big\},\mathcal{W}_{N\backslash\mathcal{T}}\big)\\
	         =& H\big( \mathcal{K}_{N,\mathcal{T}}, \big\{\{{\mathbf{Z}'}_n^{(t_{i-\ell})}, n\in\mathcal{H}_{\gamma_{i-\ell}}\},
	         {\tilde{\mathbf{Z}}}_{(\gamma_{i-\ell}-1,t_{i-\ell})},\ell\in[i-1]\big\},\mathcal{W}_{N\backslash\mathcal{T}}\big)\\
	         &-H\big( \mathcal{K}_{N,\mathcal{T}}, \big\{\{{\mathbf{Z}'}_n^{(t_{i-\ell})}, n\in\mathcal{H}_{\gamma_{i-\ell}}\},
	         {\tilde{\mathbf{Z}}}_{(\gamma_{i-\ell}-1,t_{i-\ell})},\ell\in[i-1]\big\},\mathcal{W}_{N\backslash\mathcal{T}}\big|\{{\mathbf{Z}'}_n^{(t_i)}, n\in\mathcal{H}_{\gamma_i}\} \big)=0,
	        \end{align*}
	        where the last equality holds due to Lemma \ref{lemma2}, Lemma~\ref{lemma3}, Lemma~\ref{lemma4}, and the independence of local models and random vectors.
	        \begin{lemma}\label{lemma6}
	        For $\gamma_i\in[\Gamma]$, consider ${{\mathcal{Z}}'}_{{{\tilde{\mathcal{H}}}_i}}=\{{\mathbf{Z}'}_n^{(t)},n\in{\tilde{\mathcal{H}}}_i,t\in\tilde{\mathcal{T}}\}$,  where ${\tilde{\mathcal{H}}}_i$ is a subset of $\{\mathcal{H}_j,j\in[\gamma_{i}-1]\}$ of size up to $T-1$, and $|\tilde{\mathcal{T}}|\le T-1$. Then, $I\big({\tilde{\mathbf{Z}}}_{(\gamma_i,t)};{{\mathcal{Z}}'}_{{{\tilde{\mathcal{H}}}_i}}\big)=0$.
	        \begin{proof}
	   We can consider two cases: (I) If there is at least one group like $\gamma'_i\in[\gamma_i-1]$ such that $\mathcal{H}_{\gamma'_i}\notin{\tilde{\mathcal{H}}}_i$ then we can conclude that
	      $I\big({\tilde{\mathbf{Z}}}_{(\gamma_i,t)};{{\mathcal{Z}}'}_{{{\tilde{\mathcal{H}}}_i}}\big)=0$.
	     The reason is that there is a non-empty set of honest and non-dropped users that  ${\tilde{\mathbf{Z}}}_{(\gamma_i,t)}$  includes a summation of their i.i.d. random vectors, (II) If ${\tilde{\mathcal{H}}}_i=\{\mathcal{H}_j,j\in[\gamma_i-1]\}$,  then based on ramp secret sharing we have  $I\big({\tilde{\mathbf{Z}}}_{(\gamma_i,t)};{{\mathcal{Z}}'}_{{{\tilde{\mathcal{H}}}_i}}\big)=0$.
	        \end{proof}
	    \end{lemma}
	      Similar to \eqref{eq19}, using Lemma~\ref{lemma1}, Lemma~\ref{lemma2}, Lemma~\ref{lemma4}, and  Lemma~\ref{lemma6} we can proof
	      \begin{align}\label{eq20}
	         I\big({\tilde{\mathbf{Z}}}_{(\gamma_i-1,t_i)};\mathcal{K}_{N,\mathcal{T}},\{{\mathbf{Z}'}_n^{(t_i)}, n\in\mathcal{H}_{\gamma_i}\},
	        \big\{\{{\mathbf{Z}'}_n^{(t_{i-\ell})}, n\in\mathcal{H}_{\gamma_{i-\ell}}\},{\tilde{\mathbf{Z}}}_{(\gamma_{i-\ell}-1,t_{i-\ell})},\ell\in[i-1]\big\},\mathcal{W}_{N\backslash\mathcal{T}}\big)=0.
	      \end{align}
	         Using \eqref{eq19} and \eqref{eq20}, \eqref{eq18} can be written as 
	         \begin{align*}
	            H\big(\{\mathbf{F}_{(n,t_i)}, n\in\mathcal{H}_{\gamma_i}\}\big)-H\big(\{{\mathbf{Z}'}_n^{(t_i)}, n\in\mathcal{H}_{\gamma_i}\}\big)+H\big(\mathbf{S}_{(\gamma_i-1,t_i)}\big)-H\big({\tilde{\mathbf{Z}}}_{(\gamma_i-1,t_i)}\big)\le 0, 
	         \end{align*}
	        where the last term follows from the fact that uniform variables maximize entropy.
	        Thus, $ I\big(\mathcal{W}_{N\backslash\mathcal{T}};\mathcal{M}_{\tilde{U}_i}\big|\mathcal{K}_{N,\mathcal{T}},\{\mathcal{M}_{\tilde{U}_j},j\in[i-1]\} \big)=0$. 
	       	\end{proof}
	\end{lemma}
	  


	According to Lemma~\ref{lemma7}, \eqref{25} can be written as follows.
	\begin{align}\nonumber
	    I\big(&\mathcal{W}_{N\backslash\mathcal{T}};\mathcal{M}_{\mathcal{T}},\mathcal{S}_{\Gamma}\big|\mathcal{K}_{N,\mathcal{T}} \big) \nonumber\\
	    =&\sum_{i=1}^TI\big(\mathcal{W}_{N\backslash\mathcal{T}};\mathcal{M}_{\tilde{U}_i}\big|\mathcal{K}_{N,\mathcal{T}},\{\mathcal{M}_{\tilde{U}_j},j\in[i-1]\} \big)+I\big(\mathcal{W}_{N\backslash\mathcal{T}}; \mathcal{S}_{\Gamma}\big|\mathcal{K}_{N,\mathcal{T}}, \mathcal{M}_{\mathcal{T}}\big)\nonumber\\
	    =&I\big(\mathcal{W}_{N\backslash\mathcal{T}}; \{\tilde{\mathbf{Z}}_{(\Gamma,t')},t'\in[\nu]\backslash\{\mathcal{T}\cup\mathcal{D}\}\}\big|\mathcal{K}_{N,\mathcal{T}}, \mathcal{M}_{\mathcal{T}}\big)\nonumber\\ \label{28}
	    =&H\big(\{\tilde{\mathbf{Z}}_{(\Gamma,t')},t'\in[\nu]\backslash\{\mathcal{T}\cup\mathcal{D}\}\}\big|\mathcal{K}_{N,\mathcal{T}}, \mathcal{M}_{\mathcal{T}}\big)
	    -H\big(\{\tilde{\mathbf{Z}}_{(\Gamma,t')},t'\in[\nu]\backslash\{\mathcal{T}\cup\mathcal{D}\}\}\big|\mathcal{K}_{N,\mathcal{T}}, \mathcal{R}_{\mathcal{T}},\mathcal{W}_{N\backslash\mathcal{T}}\big)=0,
	\end{align}
		where  $\mathcal{R}_{\mathcal{T}}\triangleq\bigcup_{i\in[T]}\mathcal{R}_{\tilde{U}_i}$, and 
  $\mathcal{R}_{\tilde{U}_i}~\triangleq~\big\{\{{\mathbf{Z}'}_n^{(t_i)},n\in\mathcal{H}_{\gamma_i}\},\tilde{\mathbf{Z}}_{(\gamma_i-1,t_i)}\big\}$. 
  Using argument similar to that in the proof of Lemma~\ref{lemma7}, independence of local models and random vectors, and according to Lemma~\ref{lemma1}, Lemma~\ref{lemma4} and Lemma~\ref{lemma6}, both terms in \eqref{28} are equal to $H\big(\{\tilde{\mathbf{Z}}_{(\Gamma,t')},t'\in[\nu]\backslash\{\mathcal{T}\cup\mathcal{D}\}\}\big)$ and the result is 0. Therefore, the privacy constraint is satisfied, i.e.,
  \begin{align*}
    I\big(\mathcal{W}_{N\backslash\mathcal{T}};\mathcal{M}_{\mathcal{T}},\mathcal{S}_{\Gamma},\{\mathbf{W}_k,\mathcal{Z}_k,{k\in\mathcal{T}}\}\big|\sum_{n\in[N]\backslash\{\mathcal{D}\cup\mathcal{T}\}}{\mathbf{W}_n} \big)=0,
\end{align*}
  	which proves the privacy constraint.

	\section{conclusion}
In this paper we propose \swift, which is a secure aggregation protocol for model aggregation in federated learning. Via partitioning users into groups and careful designs of intra and inter group secret sharing and aggregation schemes, \swift is able to achieve correct aggregation in presence of $D$ dropout users, with the worst-case security guarantee against $T$ users colluding with a curious server. Compared with previous secure aggregation protocols, \swift does not require an all-to-all communication network among users, and significantly slashes the communication load of the server.

	\bibliographystyle{ieeetr}
	\bibliography{References}
	
\end{document}